\documentclass[12pt]{article}

 \usepackage[abbrvbib, preprint]{jmlr2e}

\usepackage{amsmath}
\usepackage{bm}
\usepackage{bbm}
\usepackage{fancyhdr}
\usepackage{natbib}

\def\y{\mathbf{y}}

\def\I{\mathbf{I}}

\def\X{\mathbf{X}}

\def\0{\mathbf{0}}

\newcommand{\x}{ {\bf x} }

\newcommand{\bbeta}{ {\boldsymbol \beta} }

\makeatletter
\newcommand*\bigcdot{\mathpalette\bigcdot@{.5}}
\newcommand*\bigcdot@[2]{\mathbin{\vcenter{\hbox{\scalebox{#2}{$\m@th#1\bullet$}}}}}
\makeatother

\ShortHeadings{Bayesian Distance Weighted Discrimination}{Eric F. Lock}
\firstpageno{1}

\begin{document}

\title{Bayesian Distance Weighted Discrimination}

\author{\name Eric F. Lock \email elock@umn.edu \\
       \addr Division of Biostatistics\\
       School of Public Health \\
       University of Minnesota\\
       Minneapolis, MN 55446, USA
}

\date{}


\maketitle

\begin{abstract}
Distance weighted discrimination (DWD) is a linear discrimination method that is particularly well-suited for classification tasks with high-dimensional data.  The DWD coefficients minimize an intuitive objective function, which can solved very efficiently using state-of-the-art optimization techniques.  However, DWD has not yet been cast into a model-based framework for statistical inference.   In this article we show that DWD identifies the mode of a proper Bayesian posterior distribution, that results from a particular link function for the class probabilities and a shrinkage-inducing proper prior distribution on the coefficients.  We describe a relatively efficient Markov chain Monte Carlo (MCMC) algorithm to simulate from the true posterior under this Bayesian framework.  We show that the posterior is asymptotically normal and derive the mean and covariance matrix of its limiting distribution.    Through several simulation studies and an application to breast cancer genomics we demonstrate how the Bayesian approach to DWD can be used to (1) compute well-calibrated posterior class probabilities, (2) assess uncertainty in the DWD coefficients and resulting sample scores, (3) improve power via semi-supervised analysis when not all class labels are available, and (4) automatically determine a penalty tuning parameter within the model-based framework. R code to perform Bayesian DWD is available at \url{https://github.com/lockEF/BayesianDWD}. \\
\end{abstract}


\section{Introduction}
\label{intro}

High-dimensional data occur when a massive number of features are available for each unit of observation in a sample set.  Such data are now very common across a variety of application areas, and this has motivated a large number of data analysis methods that were developed specifically for the high-dimensional context.   For supervised analysis, where the task is to predict or describe an outcome, methods that are appropriate for the high-dimensional context often begin with an objective function to be optimized but lack a clear model-based framework for statistical inference.  One reason for this is because fitting statistical models directly via (e.g.) maximum likelihood is prone to over-fitting and lack of identifiability in the high-dimensional context, especially when the number of features is larger than the sample size.  Thus, objective functions that admit more regularization are required, and these often do not have a model-based interpretation in the frequentist framework. For example,  a common approach in the high-dimensional context is to incorporate penalty terms that enforce shrinkage or sparsity (e.g., $l_2$ or $l_1$ penalization) within a classical predictive model (e.g., generalized linear models). 

A Bayesian framework provides more flexibility for model-based supervised analysis of high-dimensional data, as appropriate regularization can be induced through the specified prior distribution.  Indeed, there are a number of instances in which optimizing a commonly used penalized objective function was later shown to give the mode of a posterior distribution under a particular Bayesian model.   For example, a Bayesian linear model with a normally distributed outcome and a normal prior on the coefficients corresponds to solving an $l_2$ penalized least squares objective (i.e., ridge regression)  \citep{hsiang1975bayesian}.  Similarly, the mode of the posterior under a Bayesian framework with a double-exponential prior on the model coefficients corresponds to $l_1$ penalization (i.e., the lasso) \citep{park2008bayesian}.  Such equivalences can be very useful to better understand and interpret the original optimization problem, and also provide a potential framework for statistical inference that is based on the original objective.  
 
In this article we introduce a Bayesian formulation for Distance Weighted Discrimination (DWD) \citep{marron2007distance}.  DWD is a popular approach for supervised analysis with a binary outcome and high-dimensional data.  It aims to identify a linear combination of the features that distinguish the sample groups corresponding to the binary outcomes (i.e., projections onto a separating hyperplane); in this respect, it is comparable to various versions of the Support Vector Machine (SVM) \citep{cortes1995support} or penalized Fisher's Linear Discriminant Analysis (LDA) \citep{witten2011penalized}.  However, relative to its competitors, DWD often has superior generalizability for settings in which the sample size is small relative to the dimension, i.e.,  the high dimension low sample size (HDLSS) context.  In its original formulation, DWD minimizes the inverse distances of the observed units from the separating hyperplane with a penalty term.  This circumvents the ``data piling" issue of SVM, wherein projections tend to pile up at the margins of the separating hyperplane for HDLSS data as a symptom of over-fitting and lack of generalizability. Since its inception, many other versions of DWD have been developed, including extensions that allow more than two classes \citep{huang2013multiclass}, sparsity \citep{sparseDWD}, non-linear kernels \citep{wang2019multicategory}, multi-way (i.e., tensor) data \citep{lyu2016discriminating}, and sample weighting for unbalanced data \citep{qiao2010weighted}.  These useful extensions modify the optimization task used for DWD, but remain model-free.  Resampling methods such as permutation testing \citep{wei2016direction}, cross-validation, and bootstrapping \citep{lyu2016discriminating} have been used with DWD to address inferential questions that involve uncertainty.  However, while various model-based versions of SVM have been proposed \citep{sollich2002bayesian, henao2014bayesian}, DWD is yet to be cast into a fully specified probability model.  

We show that DWD identifies the mode of a proper Bayesian posterior distribution. The corresponding density of the posterior distribution is a monotone function of the DWD objective.  The underlying model is given by a particular link function for the class probabilities and a shrinkage-inducing proper prior distribution on the model coefficients.  In addition to providing a model-based context for DWD, we demonstrate how this Bayesian framework can be used to extend and enhance present applications of DWD in four specific ways:
\begin{enumerate}
\item Posterior class inclusion probabilities can be used for ``soft" classification, and we find that these probabilities are well-calibrated across different contexts.
\item Posterior credible intervals can be used to visualize uncertainty in the coefficients and the sample scores (i.e., projections).
\item Semi-supervised analyses, in which only some of the class labels are observed in the training data, fit naturally into the Bayesian framework and can improve power.  
\item A penalty tuning parameter for the DWD objective can be subsumed within the Bayesian framework and given its own prior, to allow for model-based selection of this parameter.   
\end{enumerate} 
We describe a relatively efficient Markov chain Monte Carlo (MCMC) algorithm to simulate from the posterior distribution.  Alternatively, we show that the posterior is asymptotically multivariate normal; we derive the asymptotic mean and covariance, which are readily computable using existing software.  

In what follows, we introduce our formal framework and notation in Section~\ref{context}, formally describe the DWD objective in Section~\ref{DWD}, and then present and discuss the Bayesian model and its estimation in Sections \ref{model}, \ref{asymp.norm}, \ref{extensions}, and \ref{postcomp}. We present simulation results to illustrate and validate different aspects of the approach in Section \ref{sims}, and we present an application to subtype classification in breast cancer genomics in Section \ref{app}.   

\section{Framework and notation}
\label{context}
Throughout this article bold lower-case characters denote vectors, bold upper-case characters denote matrices, and greek characters denote unknown model parameters.  Probability density functions are denoted by $p(\cdot)$, and discrete probabilities by $P(\cdot)$.

 Let $\x_{i}$ be a $d$-dimensional vector and $y_i \in \{-1,1\}$ be a corresponding class indicator for sampled observations $i=1,\hdots,n$. Let $\X$ be the $d \times n$ matrix $\X = [\x_1 \; \x_2 \hdots \x_n ]$.  Let $\y$ be the $n$-dimensional vector of outcomes $\y=[y_1 \; y_2 \hdots y_n]$.   Our task is to predict the outcomes $y_i$ from the data $\x_i$.  

\section{Distance weighted discrimination}
\label{DWD}

 Distance weighted discrimination (DWD) identifies coefficients (weights) $\bbeta = [\beta_1 \; \beta_2 \hdots \beta_d]$ and an intercept $\beta_0 \in \mathbb{R}$ such that the linear combinations (scores) $\beta_0 + \x_{i}^T \bbeta$ discriminate the classes.   The optimization problem for DWD was originally formulated as 
\begin{equation}
\underset{\beta_0,\bbeta,\bm{\xi}}{\mbox{argmin}}\sum_{i=1}^n \frac{1}{r_i}+C\sum_{1}^{n} \xi_i, \text{ where } r_i=y_i(\beta_0+\x_i^T \bbeta)+\xi_i
\label{DWD_orig}
\end{equation}
subject to  $r_i>0$ and $\xi_i>0$ for $i=1,\hdots,n$ and $||\bbeta||_2^2 \leq 1$.  In practice the classification decision rule is determined by the sign of $\beta_0 + \x_{i}^T \bbeta$, and thus the $\xi_i$ can be interpreted as a penalty for misclassification.  The tuning parameter $C$ controls the relative weight of this misclassification penalty compared to the inverse distances from the separating hyperplane $1/r_i$. An appealing property of DWD is that its objective accounts for the projections of all the data observations, rather than just those at the margins. \cite{liu2011hard} show that the DWD objective \eqref{DWD_orig} is equivalent to the following ``loss+penalty" formulation
\begin{align}
\underset{\beta_0, \bbeta}{\mbox{argmin}} \frac{1}{n} \sum_{i=1}^n V\left(y_i(\beta_0 +\x_{i}^T \bbeta)\right)	+ \frac{\lambda}{2} ||\bbeta||^2_2, \label{dwd_obj} 
\end{align}
where $\lambda$ is a penalty tuning parameter, and 
\begin{align}
V(u) = \begin{cases}
 1-u &\text{ if } u \leq 1/2,\\
 1/(4u) &\text{ if } u > 1/2.
 \end{cases}
 \label{vu}
\end{align}
The objectives \eqref{DWD_orig} and \eqref{dwd_obj} are equivalent in that, with a one-to-one correspondence between the penalty parameters $\lambda$ and $C$, the estimated coefficients and resulting sample scores are proportional.

\section{Bayesian model}
\label{model}

Here we show that the DWD objective~\eqref{dwd_obj} gives the mode of a posterior distribution for $\{\beta_0, \bbeta\}$ under a Bayesian framework.  Let $\psi(\cdot)$ give the objective function that is minimized by DWD,
\[\psi(\beta_0,\bbeta,\y,\X,\lambda) = \frac{1}{n} \sum_{i=1}^n V\left(y_i(\beta_0 +\x_{i}^T \bbeta)\right)	+ \frac{\lambda}{2} ||\bbeta||^2_2.\]
The posterior density for $\{\beta_0, \bbeta\}$ is the following monotone transformation of $\psi$,
\begin{align}
p(\bbeta,\beta_0 \mid \X, \y, \lambda) &\propto e^{-n \psi(\beta_0,\bbeta,\y,\X,\lambda)} \nonumber \\
\label{bpost} &= \left[ \prod_{i=1}^n e^{-V\left(y_i(\beta_0 +\x_{i}^T \bbeta)\right)} \right]   e^{-(\lambda n/2) ||\bbeta||^2_2} 
\end{align}
Under general conditions \eqref{bpost} is integrable over $\{\beta_0, \bbeta\}$, and therefore defines a proper probability density function; we state this formally in Theorem \ref{thm1}.

\begin{theorem}
\label{thm1}
If $\X \in \mathbb{R}^{d \times n}$, $\lambda \geq 0$, and $\y \in \{-1,1\}^n$ where $y_i=-1$ for some $i$ and $y_j=1$ for some $j$, then $p(\bbeta,\beta_0 \mid \X, \y, \lambda)$ in \eqref{bpost} gives a proper probability density function over $\{\beta_0, \bbeta\}$.      	
\end{theorem}

In what follows, we ``work backward" from this posterior to complete the specification of the model and discuss its implications.  

\subsection{Conditional distribution of $\y$}
\label{ydist}

Treating $\y$ as a discrete random vector, it follows directly from~\eqref{bpost} that its entries $y_i$ are conditionally independent and that 
\[P(y_i \mid \bbeta, \beta_0, \x_i) \propto P(y_i) e^{-V\left(y_n(\beta_0 +\x_{i}^T \bbeta)\right)}.\]
The conditional distribution for $y_i \in \{-1,1\}$ thus depends on its unconditional class probability $P(y_i=1)=1-P(y_i=-1)$, and its corresponding score $u_i := \beta_0 +\x_{i}^T \bbeta$:
\begin{align}
P(y_i = 1 \mid u_i) = \frac{P(y_i=1) e^{-V\left(u_i\right)}}{P(y_i=1) e^{-V\left(u_i\right)}+P(y_i={-1})e^{-V\left(-u_i\right)}}.	 \label{ydist}
\end{align}
Choice of $P(y_i=1)$ does not influence posterior inference for $\{\beta_0, \bbeta\}$ in  \eqref{bpost}, but it is nevertheless important for the calibration of the predictive model and is discussed further in Section~\ref{uneq.class}.  Given $P(y_i=1)$, \eqref{ydist} defines a link function from $u_i \in \mathbb{R}$ to a probability in $[0,1]$; this function is smooth and comparable to a probit or logit link. Figure~\ref{links} plots this function under equal class proportions, $P(y_i=1)=1/2$.  

\begin{figure}[!ht]
\centering
\includegraphics[width=0.9\textwidth]{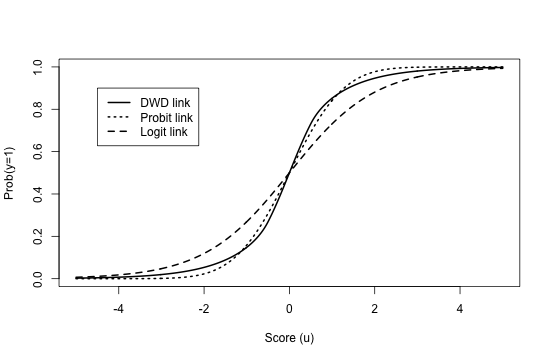}	
\caption{Class probability as a function of linear score under different links.}
\label{links}
\end{figure}

\subsection{Prior distribution of $\beta$}
\label{bdist}
Estimating $\beta_0$ and $\bbeta$ via maximizing the likelihood under model~\eqref{ydist} would perform poorly in the HDLSS context, analogous to the poor performance of unconstrained probit or logistic regression models.  However, the posterior distribution~\eqref{bpost} also implies a ``prior" distribution for $\{\beta_0,  \bbeta\}$ that is not conditional on $\y$:
\begin{align}
 p(\bbeta,\beta_0 \mid \X, \y, \lambda)  &\propto  P(\y \mid \bbeta, \beta_0, \X) \, p(\bbeta, \beta_0 \mid \X, \lambda) \nonumber \\
 \rightarrow  p(\bbeta, \beta_0 \mid \X, \lambda) & \propto A \cdot B \label{bprior}
\end{align}
where 
\begin{align} A = \prod_{i=1}^n \left( P(y_i=1)e^{-V\left(\beta_0 +\x_{n}^T \bbeta\right)} +P(y_i=-1)e^{-V\left(-\beta_0 -\x_{n}^T\bbeta\right)}\right) \label{termA} \end{align}
and
\begin{align} B = e^{-(\lambda N/2) ||\bbeta||^2_2}. \label{termB} \end{align}
This prior facilitates shrinkage in two important ways that improve performance in HDLSS settings.  
Term $B$~\eqref{termB} gives the kernel of a multivariate normal distribution for $\bbeta$, in which the $\beta_i$'s are independent with mean $0$ and variance $1/(\lambda N)$.  This derives from the $l_2$ penalty in~\eqref{dwd_obj}.  There is a vast literature on the connection between $l_2$-regularized regression (e.g., \emph{ridge regression}) and using Gaussian priors on the coefficients in a Bayesian framework, as mentioned in Section~\ref{intro}.  This shrinks the inferred coefficients toward $\0$, which improves performance in the HDLSS setting or in the presence of multicollinearity.  Term $A$ \eqref{termA} favors $\bbeta$ that correspond to directions in $\X$ with high variability or bimodality.  To illustrate, Figure~\ref{prior_fig} plots a single term of the product $A$ as a function of the DWD score $u_i = \beta_0+\x_i^T \bbeta$ when $P(y_i=1)=1/2$,
\begin{align} f(u_i)=e^{-V(u_i)}+e^{-V(-u_i)}. \label{prior_term} \end{align}
\begin{figure}[!ht]
\centering
\includegraphics[width=0.9\textwidth]{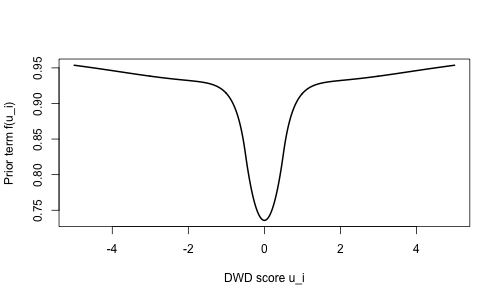}	
\caption{Prior term $f(u_i)$ \eqref{prior_term} as a function of DWD score $u_i$.}
\label{prior_fig}
\end{figure}
This term gives higher prior density to directions in $\X$ for which the corresponding scores $\{u_i\}_{i=1}^N$ are not concentrated near $0$, and are therefore more useful for discriminating a negative class from a positive class.     
  
The function $A \cdot B$ in~\eqref{bprior} defines a proper density for $\bbeta$ for any fixed intercept $\beta_0$, as stated in Theorem \ref{thm2}.
\begin{theorem}
\label{thm2}
If $\X \in \mathbb{R}^{d \times n}$, $\lambda > 0$, and $\beta_0 \in  \mathbb{R}$,  then \[p(\bbeta \mid \X, \beta_0, \lambda) \propto A \cdot B\]
gives a proper probability density function, where $A$ and $B$ are defined in~\eqref{termA} and~\eqref{termB}.   	
\end{theorem}
The prior \eqref{bprior} does not define a proper density for $\beta_0$.  Improper priors are commonly used for a variety of Bayesian applications, and still provide a valid framework for inference if the posterior is proper \citep{taraldsen2010improper,sun2001propriety}.  Nevertheless, it would be straightforward to modify the model to allow for shrinkage of $\beta_0$, if one desired it to have a proper probability density without conditioning on $\y$.            
  
 \subsection{Distribution of $X$}
\label{xdist}
No explicit probabilistic assumptions on $\X$ are needed for the validity of the posterior and prior distributions in the previous two sections.  However, it is informative to consider the broader implications of the model with respect to the distribution for $\X$.  Without loss of generality we set $\beta_0=0$ and assume $P(y_i=1)=0.5$.  Note that
\[p(\bbeta \mid \x_i, \lambda) = \left( e^{-V\left(\x_{i}^T \bbeta\right)}+e^{-V\left(-\x_{i}^T\bbeta\right)} \right) e^{-(\lambda/2) ||\bbeta||^2_2} / \phi(\x_i, \lambda), \]
where
\[\phi(\x_i, \lambda) = \int_{\bbeta}  \left( e^{-V\left(\x_{i}^T \bbeta\right)}+e^{-V\left(-\x_{i}^T\bbeta\right)} \right) e^{-(\lambda/2) ||\bbeta||^2_2} d \bbeta, \]
Given any marginal density for $\x$ (which may depend on $\lambda$), $\x_i \sim p(\x \mid \lambda)$, 
\begin{align*}
p(\bbeta, \x_i \mid \lambda) = p(\bbeta \mid \x_i, \lambda) p(\x_i), \\
p(\bbeta \mid \lambda) = \int p(\bbeta, \x_i \mid \lambda) d \x_i, \\
p(\x_i \mid \bbeta, \lambda) = p(\bbeta, \x_i \mid \lambda)/p(\bbeta \mid \lambda).
\end{align*}
Thus, a marginal prior for $\bbeta$ exists but it depends on the distribution $\x_i$.   
Moreover,
\begin{align*}
p(\x_i \mid \bbeta, \lambda) = p(\x_i \mid \bbeta, y_i=1, \lambda) + p(\x_i \mid \bbeta, y_i=-1, \lambda)  	
\end{align*}
implies that 
\[p(\x_i \mid \bbeta, y_i) \propto e^{-V(y_i \x_i^T \bbeta)} \phi(\x_i,\lambda) p(\x_i\mid \lambda) / \phi(\x_i,\lambda).\]
Thus, given any marginal distribution for the data $\x_i$, it is possible to obtain the conditional distribution of $\x_i$ given its class membership $\y_i$ and the DWD vector $\bbeta$.   

  \section{Asymptotic normality of $\beta$}
\label{asymp.norm}

It follows from the Bernstein-von Mises theorem (i.e., the ``Bayesian central limit theorem") \citep{ghosal1997review} that the posterior distribution for $\bbeta$ approximates a multivariate normal distribution as $n \rightarrow \infty$.  The limiting distribution can be derived precisely via a second-order Taylor expansion of the log of \eqref{bpost}, about the mode $\hat{\bbeta}$.  The resulting mean and covariance matrix are provided in Theorem~\ref{thm3}.
\begin{theorem}
\label{thm3}
If $\X \in \mathbb{R}^{d \times n}$, $\lambda > 0$, and $\y \in \{-1,1\}^n$, then  
\[p(\bbeta \mid \X, \y, \lambda, \beta_0) \approx \mbox{MVN}(\bbeta \mid \hat{\bbeta}, V_\beta)\]
as $n\rightarrow \infty$, where $\hat{\bbeta}$ is the DWD solution \eqref{dwd_obj}, and 
\[V_\beta = \left( \sum_{i \in S_{\hat{\bbeta}}} \frac{\x_i \x_i^T}{2(y_i (\hat{\beta}_0+\x_i^T \hat{\bbeta})^3} +n\lambda \mathbf{I} \right)^{-1}\] 	
where $S_{\hat{\bbeta}} = \{i : y_i (\hat{\beta}_0+\x_i^T \hat{\bbeta}) > 0.5\}$.
\end{theorem}
Note that this asymptotic approximation always exists and is directly computable given the DWD solution $\{\hat{\beta_0},\hat{\bbeta} \}$.

\section{Extensions}
\label{extensions}

\subsection{Marginal class probabilities}
\label{uneq.class}

The prior class probabilities $P(y_i=1)$ and $P(y_i=-1)$ in \eqref{ydist} do not affect the posterior for $\{\beta_0, \bbeta\}$ in~\eqref{bpost}.  Thus, they need not be specified for tasks that do not require a predictive model for $y_i$, such as exploratory visualization or batch adjustment \citep{huang2012r}.  However, their specification is important to obtain appropriately calibrated probabilities for classification.  In particular, they must be considered carefully when the proportions in the training data differ from that of the target population.  The ability to adjust the marginal class probabilities for out-of-sample prediction is useful, as the decision boundary for DWD ($\beta_0$) is sensitive to the class proportions used for training \citep{qiao2010weighted,qiao2015flexible}.
We let each observation have the same class membership probability \emph{a priori}:
\[P(y_i=1)=P_1  \; \text{ for } i=1,\hdots,n,\]
and by default we set $P_1=0.5$.  In other situations $P_1$ may be known and fixed at another value, or estimated as the sample proportion from class $1$. Alternatively, one can give $P_1$ its own prior (e.g., $P_1 \sim \text{Uniform}(0,1)$) and infer it with the rest of the Bayesian model; such an approach may be useful in (e.g.,) situation for which the $y_i$'s are only partially observed and the population class proportions are of interest. 

\subsection{Inference for $\lambda$}
\label{lambda.sec}

Rather than fixing the penalty parameter $\lambda$ \emph{a priori}, we recommend inferring its value within the Bayesian framework.  For a given prior density $\lambda \sim p(\lambda)$, the conditional density for $\lambda$ is 
\begin{align}
\begin{split}
	p(\lambda \mid \X,\y,\bbeta,\beta_0) &\propto p(\lambda) P(\y \mid \X, \bbeta, \beta_0) p(\bbeta,\X \mid \beta_0, \lambda)  \\
&\propto \prod_{i=1}^n p(\lambda) \left( P_1e^{-V\left(\beta_0 +\x_{i}^T \bbeta\right)} +(1-P_1)e^{-V\left(-\beta_0 -\x_{i}^T\bbeta\right)}\right)/ \phi(\X,\beta_0,\lambda) 
\end{split}
\label{lambda.cond}
\end{align}

where
\begin{align}\phi(\X, \lambda,\beta_0) = \int_{\bbeta}  \prod_{i=1}^n \left( P_1 e^{-V\left(\beta_0+\x_{i}^T \bbeta\right)}+(1-P_1)e^{-V\left(-\beta_0-\x_{i}^T\bbeta\right)} \right) e^{-(\lambda/2) ||\bbeta||^2_2} d \bbeta.	 \label{phi.label}
\end{align}

In practice we use a prior for $\lambda$ that is uniform over a large range, by default, $\lambda \sim \text{Uniform}(1/128,128)$.  The impact of $\lambda$ depends on the scale of $\X$ and the number of observations $n$, so the prior may need to be suitably modified for other scenarios.

\subsection{Semi-supervised learning}
\label{semisup}

The conditional distribution of $\bbeta$ on $\X$ in \eqref{bprior} is not flat, and therefore observations for which $y_i$ are not observed can still inform $\bbeta$.  Consider the context in which $n_o$ observations are fully observed $\{(\x_i,\y_i)\}_{i=1}^{n_o}$,  $n_u$ have missing outcome $\{(\x_i^u,y_i)\}_{i=1}^{n_u}$, and $n=n_o+n_u$.  The full posterior distribution for $\{\beta_0,\bbeta\}$ is
\begin{align*}
p(\bbeta,\beta_0 \mid \X, \y, \lambda) &\propto \left[ \prod_{i=1}^{n_u} \left( P_1e^{-V\left(\beta_0 +\x_{n}^T \bbeta\right)} +(1-P_1)e^{-V\left(-\beta_0 -\x_{n}^T\bbeta\right)}\right) \right] \\ & \; \; \; \; \; \; \left[ \prod_{i=1}^{n_o} e^{-V\left(y_i(\beta_0 +\x_{i}^T \bbeta)\right)} \right]   e^{-(\lambda n/2) ||\bbeta||^2_2}.   
\end{align*}
This full posterior can be used for semi-supervised learning.


\section{Posterior approximation}
\label{postcomp}

\subsection{MCMC algorithm, fixed $\lambda$}
\label{mh.alg}

To simulate from the posterior distribution for $\{\beta,\bbeta_0\}$ \eqref{bpost}, we implement a ``Metropolis-in-Gibbs" algorithm \citep{carlin2008bayesian}.  That is, we iteratively draw each $\beta_i$ from its conditional distribution given the current values of the rest of the parameters $\bbeta[-i]$, $p(\beta_i \mid \bbeta[-i], \X, \y, \lambda)$; each conditional draw is accomplished via a Metropolis sub-step.  For the Metropolis step, we generate proposals from a normal distribution with mean given by the previous value for $\beta_i$ and a variance that is adaptively scaled over the course of the algorithm.  We initialize the Markov chain at the posterior mode, which can be identified using existing software for DWD;  we use the {\tt sdwd} package for this initialization.  

\subsection{Inference for $\lambda$}
\label{lambda.inf}

Inference for $\lambda$ can be accomplished by iteratively simulating from its full conditional distribution \eqref{lambda.cond} to update its value within the Gibbs sampling algorithm in Section~\ref{mh.alg}.  However, this is not straightforward, because the integral $\phi(\X, \beta_0, \lambda)$ \eqref{phi.label} does not have a closed form.  In our implementation, we first estimate $\phi(\X,\lambda,\beta_0)$ over a grid of values $\{\lambda_j:j=1,\hdots,J\}$ that span the range of the prior for $\lambda$.  For this, we utilize the fact that $e^{-(\lambda n/2) ||\bbeta||^2_2}$ resembles a Gaussian kernel to approximate the integral via Monte Carlo with simulated random normal vectors.  Given $\bbeta^{(t)}\sim \mbox{MVN}(\mathbf{0},1/(\lambda n) \I)$ for $t=1,\hdots,T$,
       \[\phi(\X,\lambda,\beta_0) \approx \left(\frac{2 \pi}{n \lambda}\right)^{p/2} \frac{1}{T} \sum_{t=1}^T   \prod_{i=1}^n \left( P_1 e^{-V\left(\beta_0+\x_{i}^T \bbeta^{(t)}\right)}+(1-P_1)e^{-V\left(-\beta_0-\x_{i}^T\bbeta^{(t)}\right)} \right).\]
      Using this approximation, we estimate $\hat{\phi}(\X,\lambda_j,\beta_0)$ for each $j$, and then linearly interpolate on the log-scale to approximate the full function $\hat{\phi}(\X,\lambda,\beta_0)$ over the range of the prior for $\lambda$.  Posterior estimation then proceeds as in Section \ref{mh.alg}, with an additional Metropolis sub-step to update the value of $\lambda$.  The full algorithm for posterior sampling is given in Appendix~\ref{alg}.
      
  \subsection{Computational considerations}
  \label{computing}
  
  In practice, we find that $1000$ MCMC cycles yield appropriate coverage of the posterior when $\lambda$ is fixed (see Section~\ref{genmodsim}) and $10000$ are sufficient when $\lambda$ is inferred.  On a single CPU, $10000$ MCMC iterations takes approximately $5$ minutes for a dataset of size $\X: 500 \times 500$ (comparable to the data considered in Section~\ref{app}) and computing time scales linearly with the dimension $d$.  Thus, full posterior inference via MCMC is feasible for most applications (with patience), but computing time can be a limitation.        

\section{Simulations}
\label{sims}

We consider three distinct simulation scenarios, and multiple conditions under each scenario, to illustrate and validate several aspects of Bayesian DWD.  

In Section~\ref{genmodsim} we simulate from a data generating scheme that matches the assumed prior and likelihood exactly.  Under this scenario, we confirm that posterior inferences using the MCMC algorithms in Section~\ref{postcomp} are appropriately calibrated regardless of the marginal distribution for $\X$. We compare to results using the normal approximation in Section~\ref{asymp.norm} and a bootstrapping approach.

In Section~\ref{mvn.simulation} we consider a scenario in which the two classes are defined by different multivariate normal distributions.  Under this scenario we demonstrate the robustness of Bayesian DWD to its model assumptions, assess the sensitivity of results to the tuning parameter $\lambda$, and validate model-based inference for $\lambda$.

In Section~\ref{semisup.sim} we consider different scenarios for which not all of the outcomes are observed to illustrate potential benefits and limitations of the semi-supervised approach.      

\subsection{Assumed model simulation}
\label{genmodsim}

Here we fix $\beta_0=0$ and simulate data from the assumed model as follows:
\begin{enumerate} 
\item Generate $\X: d \times n$ according to a specified probability distribution,
\item Generate $\bbeta_{\text{true}}$ from its conditional prior $\eqref{bprior}$ (via a Metropolis-Hastings algorithm), given $\lambda$
\item For $i=1,\hdots,n$, compute the score $u_i=\x_i^T \bbeta_{\text{true}}$ and generate $y_i$ according to its class probabilities defined by \eqref{ydist}.
\end{enumerate}

As manipulated conditions, we consider three values for $\lambda$ ($\lambda=0.1,1,$ or $10$), two configurations of $d$ and $n$ ($\X: 20 \times 100$ or $\X: 100 \times 20$), and three different probability distributions for $\X$: (1) entries are simulated independently from a Uniform$(-1,1)$ distribution, (2) entries are simulated independently from an Exponential$(1)$ distribution centered to have mean $0$, and (3) each column $\x_i$ is simulated from the following mixture of multivariate normal distributions:
\begin{align}x_i \sim \begin{cases}\mbox{MVN}(\bm{\mu}_0, \I) \text { with probability } 1/2 \\
 \mbox{MVN}(\bm{\mu}_1, \I) \text { with probability } 1/2,\end{cases} \label{bimod.eq} \end{align}
 where the entries of $\bm{\mu}_0$ and $\bm{\mu}_1$ are generated independently from a N$(0,0.5)$ distribution.  
 
For each set of manipulated conditions, we generate $100$ datasets $\{\X, \bbeta_{\text{true}}, \y\}$ under the assumed model and approximate the posterior distribution for $\bbeta$ using the MCMC algorithm in Section~\ref{mh.alg}. We construct 95\% Bayesian credible intervals for $\bbeta_{\text{true}}$ and the scores $\mathbf{u}= \X^T \bbeta_{\text{true}}$ based on the posterior samples.  We also consider credible intervals constructed using the posterior mode and the asymptotic approximation in Theorem~\ref{thm3} (i.e., the CLT approximation), and confidence intervals based on the percentiles of $200$ non-parametric bootstrap resamples as in \citet{lyu2016discriminating}. 
   
 Table~\ref{tab.cov} gives the average coverage rates for the true values under each set of conditions and the three different approaches to inference.  The credible intervals using MCMC all have nominal coverage rates, which validates the algorithm.  The CLT approximation gives appropriate coverage in several instances, even when $p>n$.  However, it performs less well when $\lambda$ is smaller and has relatively lower coverage rates for the scores $\mathbf{u}= \X^T \bbeta_{\text{true}}$.  The bootstrap percentile approach yields more narrow intervals and tends to undercover in all cases.  As a representative example, Figure~\ref{fig:cov1} plots the mode and 95\% intervals for the coefficients for an instance under the Uniform scenario with $\lambda=0.1$, $n=100$ and $p=20$; this shows close agreement between the MCMC and CLT intervals, and the bootstrap intervals are universally more narrow.
 
To assess the calibration of estimated class probabilities, we generate a new set of $1000$ observations under the true model as a ``test" set for each instance of the simulation.  We compute the class probabilities for each test case, given the posterior fit to the $n=20$ or $n=100$ training cases.  We group the estimated probabilities into narrow bins of length $0.02$ ($[0,0.02), [0.02,0.04)$, etc.) and consider the proportion of cases with $y=1$ in each bin.  Figure~\ref{fig:calib} plots these proportions, aggregated across all conditions, vs.\ the midpoint of each bin.  Using the posterior mean, the empirical proportions match the estimated probabilities perfectly, confirming that the model is well-calibrated and again validating the MCMC procedure. Using the posterior mode, the relationship deviates slightly but still tends to provide a reasonable fit in most cases.   
 
\begin{table}
\caption{\label{tab.cov} Coverage rates for model coefficients ($\beta$) and scores ($u$) using 95\% credible intervals under MCMC sampling from the posterior, the asymptotic normal approximation (CLT), and bootstrapping (BOOT).}
\centering
\begin{tabular}{cccc|cc@{\hskip 0.3in}cc@{\hskip 0.3in}cc}
  \hline
  &&& &\multicolumn{2}{c@{\hskip 0.3in}}{MCMC} & \multicolumn{2}{c@{\hskip 0.3in}}{CLT} & \multicolumn{2}{c}{BOOT}  \\
Distribution& $n$&$d$&$\lambda$& $\beta$ &$u$ & \;\;$\beta$ & $u$ &$\beta$&$u$ \\ 
  \hline
Uniform &$20$ & $100$ & $0.1$ & 0.94 & 0.94 & 0.95 & \textbf{0.85} & \textbf{0.16} & \textbf{0.42} \\ 
 Uniform & $20$ & $100$ & $1$ & 0.94 & 0.94 & 0.95 & \textbf{0.85} & \textbf{0.23} & \textbf{0.50} \\ 
  Uniform &$20$ & $100$ & $10$ & 0.94 & 0.94 & 0.95 & 0.96 & \textbf{0.20} & \textbf{0.45} \\ 
 Uniform &  $100$ & $20$ & $0.1$ & 0.94 & 0.94 & 0.92 & 0.85 & \textbf{0.46} & \textbf{0.55} \\ 
 Uniform & $100$ & $20$ & $1$ & 0.95 & 0.95 & 0.95 & 0.95 & \textbf{0.53} & \textbf{0.65} \\ 
Uniform &  $100$ & $20$ & $10$& 0.94 & 0.94 & 0.95 & 0.95 & \textbf{0.24} & \textbf{0.32} \\
  \hline 
 Exponential & $20$ & $100$ & $0.1$& 0.96 & 0.94 & \textbf{0.92} & \textbf{0.91} & \textbf{0.61} & \textbf{0.69} \\ 
  Exponential & $20$ & $100$ & $1$ & 0.95 & 0.95 & 0.96 & 0.95 & \textbf{0.62} & \textbf{0.71} \\ 
  Exponential & $20$ & $100$ & $10$ & 0.94 & 0.94 & 0.96 & 0.96 & \textbf{0.34} & \textbf{0.41} \\
  Exponential & $100$ & $20$ & $0.1$&  0.94 & 0.95 & \textbf{0.92} & \textbf{0.85} & \textbf{0.39} & \textbf{0.49} \\ 
  Exponential & $100$ & $20$ & $0.1$& 0.94 & 0.95 & 0.94 & \textbf{0.91} & \textbf{0.53} & \textbf{0.63} \\ 
  Exponential & $100$ & $20$ & $1$&  0.94 & 0.94 & 0.95 & 0.96 & \textbf{0.40} & \textbf{0.52} \\ 
   \hline 
    Bimodal & $20$ & $100$ & $10$& 0.95 & 0.96 & 0.94 & \textbf{0.91} & \textbf{0.61} & \textbf{0.66} \\ 
  Bimodal & $20$ & $100$ & $1$ & 0.94 & 0.95 & 0.96 & 0.95 & \textbf{0.66} & \textbf{0.73} \\ 
  Bimodal & $20$ & $100$ & $10$ & 0.93 & 0.94 & 0.96 & 0.96 & \textbf{0.39} & \textbf{0.46} \\ 
  Bimodal & $100$ & $20$ & $0.1$& 0.94 & 0.94 & \textbf{0.92} & \textbf{0.85} & \textbf{0.37} & \textbf{0.47} \\ 
  Bimodal & $100$ & $20$ & $1$& 0.95 & 0.95 & 0.94 & \textbf{0.90} & \textbf{0.51} & \textbf{0.60} \\ 
  Bimodal & $100$ & $20$ & $10$& 0.94 & 0.95 & 0.96 & 0.96 & \textbf{0.40} & \textbf{0.61} \\ 
 \hline
\end{tabular}
\end{table}

\begin{figure}[!ht]
\centering
\includegraphics[width=1\textwidth]{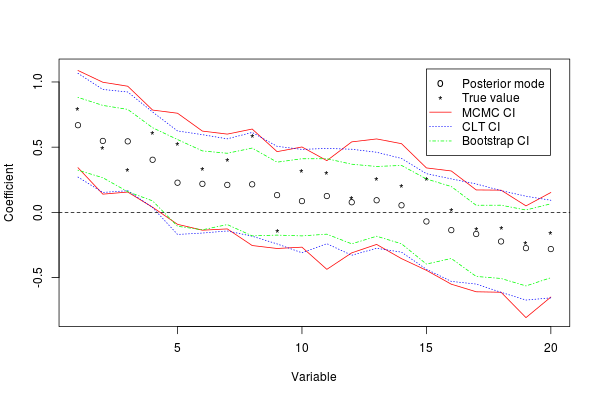}	
\caption{Inferential results and true values for the coefficients $\beta$ for a dataset generated under the Uniform scenario when $\lambda=0.1$, $n=100$ and $p=20$.}
\label{fig:cov1}
\end{figure}

\begin{figure}[!ht]
\centering
\includegraphics[width=1\textwidth]{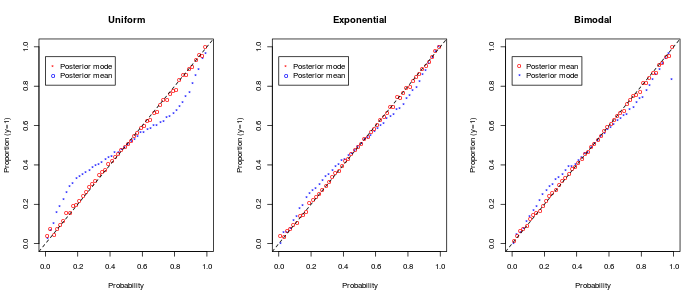}	
\caption{Estimated probabilities vs.\ observed proportions on test observations, aggregated across all conditions, for different simulation scenarios using the posterior mean or mode.}
\label{fig:calib}
\end{figure}

\subsection{Two-class multivariate normal simulation}
\label{mvn.simulation}

Here we generate data according to a two-class multivariate normal distribution, in which each class has a different mean vector, as follows:
\begin{enumerate}
\item Generate $d$-dimensional mean vectors $\bm{\mu}_0$ and $\bm{\mu}_1$, with independent entries from a N$(0,\tau^2)$ distribution.
\item Set $y_i=-1$ for $i=1,\hdots,n/2$ and $y_i=1$ for $i=n/2+1,\hdots,n$.
\item For $i=1,\hdots,n$, generate $\x_i$ via \[x_i \sim \begin{cases}\mbox{MVN}(\bm{\mu}_0, \I) \text { if } y_0=-1 \\
 \mbox{MVN}(\bm{\mu}_1, \I) \text { if } y_i=1 .\end{cases}\]
\end{enumerate}
Note that this scenario differs from the `Bimodal' case in Section~\ref{genmodsim} \eqref{bimod.eq}; for that scenario, membership in a given mixture component (defined by $\bm{\mu}_0$ and $\bm{\mu}_1$) does not necessarily correspond to the supervised class labels $y_i$.  The present scenario does not precisely match our assumed model.  The ``true" (i.e., oracle) probabilistic link between $y_i$ and $\x_i$ in this scenario corresponds to a probit link with coefficients that are proportional to the mean difference vector $\bm{\mu}_0-\bm{\mu}_1$; as shown in Figure~\ref{links} the link for Bayesian DWD is not a probit, but it is close. 

 We fix $n=100$ and consider different levels of signal ($\tau=0,0.1,0.2,0.5$) and different dimensions ($d=10,100,1000$) as manipulated conditions.  For each instance, we approximate the posterior separately with $\lambda$ fixed at each of $15$ possible values, that range from $1/128$ to $128$ and are equally spaced on a log scale.  For each case we compute the Kullback-Leibler divergence (KL-divergence) of the estimated class probabilities from the oracle model. For each case we also generate a test set of $5000$ samples and evaluate (1) the misclassification rate on the test samples using a posterior predictive probability of $0.5$ as a threshold, and the mean squared error (MSE) of the predicted class probabilities for class 1 ($p_i$) from the true class memberships as 
 \[\text{MSE}=\frac{1}{n_{\text{test}}}\sum_{i=1}^{n_{\text{test}}} (1-p_i)^2 \mathbbm{1}_{\{y_i=-1\}}+p_i^2 \mathbbm{1}_{\{y_i=1\}}.\]
Figure~\ref{fig:lamchoice} plots these metrics as a function of $\lambda$ for the different conditions.  In general, lower values of $\lambda$ perform better when the distinguishing signal is stronger.  This is intuitive, as $\lambda$ can be considered an inverse scale parameter for $\bbeta$, and larger magnitudes for $\bbeta$ correspond to probabilities closer to $0$ and $1$.  The misclassification rate is robust to the value of $\lambda$, but its specification is important to achieve well-calibrated probabilities.

\begin{figure}[!ht]
\centering
\includegraphics[width=1\textwidth]{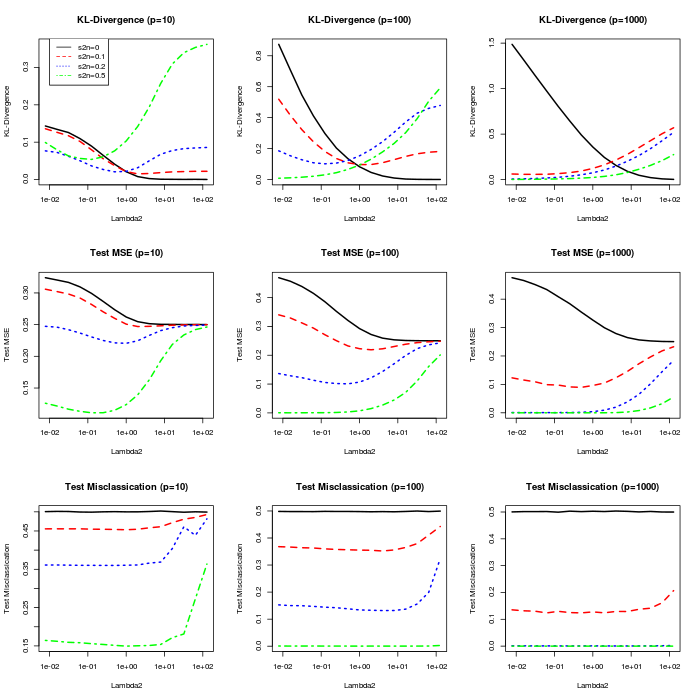}	
\caption{Different performance metrics (KL-divergence, test MSE, and test misclassification rate) are shown as a function of $\lambda$ used for estimation under different conditions.}
\label{fig:lamchoice}
\end{figure}

 Motivated by the sensitivity of the estimated probabilities to $\lambda$, we also approximate the posterior with a uniform prior on $\lambda$ as in Section~\ref{lambda.sec}.  Table~\ref{tab:lamestperm} summarizes the posterior mean for $\lambda$ under each set of conditions, and compares the performance of the model with uniform prior on $\lambda$ with that of the best performing fixed value of $\lambda$ for each metric.  The model with a uniform prior performs comparatively well across the different scenarios, and the posterior mean for $\lambda$ appropriately ranges from very high values when the signal is nonexistent to low values when the signal is strong.  However, for the conditions with $p=1000$ and strong signal the posterior mean is substantially larger than the best performing values of $\lambda$; this is likely because the model fits well and perfectly classifies the data with probabilities close to $0$ or $1$ for a wide range of $\lambda$ values (see Figure~\ref{fig:lamchoice}), so there is not strong evidence to distinguish between different values of $\lambda$ in this range.     

\begin{table}
\caption{\label{tab:lamestperm} Comparison of model performance when $\lambda$ is inferred with a uniform prior ($\hat{\lambda}$) vs.\ the best performing fixed value of $\lambda$ based on KL-divergence ($\lambda_{\text{KL}}$) or MSE $\lambda_{\text{MSE}}$.  Results for each set of condition are averages over $10$ replications.}
\centering
\begin{tabular}{cc|ccccccc}
  \hline
 $d$ & s2n ($\tau)$& $\hat{\lambda}$ & $\lambda_{\text{KL}}$ & $\lambda_{\text{MSE}}$ & KL($\hat{\lambda}$) & KL($\lambda_{\text{KL}}$)  & MSE($\hat{\lambda}$)  & MSE($\lambda_{\text{MSE}}$)  \\ 
  \hline
  \hline
$10$ & $0$ & 70.21 & 128.00 & 128.00 & 0.00 & 0.00 & 0.25 & 0.25 \\ 
 $10$ & $0.1$ & 70.82 & 2.00 & 2.00 & 0.02 & 0.02 & 0.25 & 0.25 \\ 
 $10$ & $0.2$ & 31.27 & 0.50 & 1.00 & 0.06 & 0.02 & 0.24 & 0.22 \\ 
$10$ &  $0.5$ & 0.10 & 0.12 & 0.12 & 0.06 & 0.05 & 0.11 & 0.11 \\ 
   \hline
100&$0$ & 74.56 & 128.00 & 64.00 & 0.00 & 0.00 & 0.25 & 0.25 \\ 
 100& $0.1$ & 22.55 & 1.00 & 2.00 & 0.15 & 0.10 & 0.24 & 0.22 \\ 
 100& $0.2$ & 0.14 & 0.12 & 0.50 & 0.11 & 0.10 & 0.11 & 0.10 \\ 
 100& $0.5$ & 0.11 & 0.01 & 0.02 & 0.03 & 0.01 & 0.00 & 0.00 \\ 
   \hline
   1000 & $0$& 44.79 & 128.00 & 128.00 & 0.02 & 0.00 & 0.25 & 0.25 \\ 
  1000 & $0.1$ & 3.11 & 0.03 & 0.50 & 0.21 & 0.06 & 0.11 & 0.09 \\ 
  1000 & $0.2$ & 0.63 & 0.01 & 0.03 & 0.07 & 0.01 & 0.00 & 0.00 \\ 
  1000 & $0.5$ & 1.82 & 0.01 & 0.01 & 0.04 & 0.00 & 0.00 & 0.00 \\ 
 \hline 
\end{tabular}
\end{table}

In Figure~\ref{fig:calibGauss} we consider the overall calibration of the fitted probabilities across the different conditions, as in Section~\ref{genmodsim}.  This shows clearly that fixing $\lambda$ at a large value ($\lambda=128$) can yield overly conservative estimated probabilities and a relatively small value ($\lambda=1/128$) can yield overly confident probabilities.  The model in which $\lambda$ is inferred with a uniform prior is mostly well-calibrated but does show some deviations from the estimated probabilities and empirical proportions.  These deviations are likely in part due to the fact that the assumed probabilistic link does not perfectly match the true generative model.  

\begin{figure}[!ht]
\centering
\includegraphics[width=0.8\textwidth]{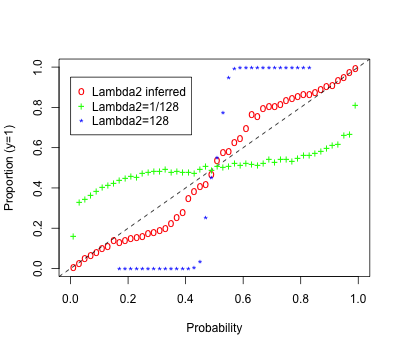}	
\caption{Estimated probabilities vs.\ observed proportions on test observations, aggregated across all conditions, for $\lambda=1/128$, $\lambda=128$, or inferring $\lambda$ with a uniform prior.}
\label{fig:calibGauss}
\end{figure}

\subsection{Semi-supervised simulation}
\label{semisup.sim}

Here we generate data in which some of the class labels $y_i$ may be unknown, to illustrate a semi-supervised approach.  We consider two scenarios.  For the first scenario, the observations $\x_i$ are generated exactly as in Section~\ref{mvn.simulation} with $d=100$ and $\tau=0.3$.  For the second scenario, $d=100$ and entries of each $\x_i$ are generated independently from a $N(0,1)$ distribution; then class labels $y_i$ are defined deterministically by the sign of $\x_i^T \bbeta_{\text{sep}}$, where the entries of $\bbeta_{\text{sep}}$ are generated from a standard normal distribution.  The key difference between these two scenarios, for our purpose, is that the first scenario (bimodal) has an underlying cluster structure that distinguishes the two classes whereas the second (unimodal) does not. So, we expect to improve performance by considering unlabeled data for the bimodal scenario but not the unimodal scenario.

As manipulated conditions, we vary the number of observations with $y_i$ observed ($n_o=10, 20, 100, 1000)$ and with $y_i$ unobserved ($n_u=0, 10, 20, 100, 1000$) in our training set, for each of the bimodal and unimodal scenarios.  For each combination of conditions, we generate $20$ datasets, use MCMC to approximate the posterior, and consider the resulting misclassification rates for a test set.  The average test misclassification rates under each set of conditions are shown in Table~\ref{tab:semisup}.  This shows that for the bimodal scenario a semi-supervised analysis with a large number of unlabeled observations can improve performance dramatically; the misclassification rates for $n_o=10$ observed and $n_u=1000$ unobserved are comparable to that with  $n_o=1000$ observed.  As expected, including data without observed class labels does not help in the unimodal scenario; however, it does not lead to a substantial decrease in performance.  

\begin{table}
\caption{\label{tab:semisup} Test misclassification rates using Bayesian DWD with differing number of observations with $y_i$ observed ($n_o$) or unobserved ($n_u$).  }
\centering
\begin{tabular}{cc|ccccc}
  \hline
   & & \multicolumn{5}{c}{$N$ unobserved ($n_u$)} \\
 Scenario & $N$ observed ($n_o$) & $0$ & $10$ & $20$ & $100$ & $1000$  \\ 
  \hline
  \hline
Bimodal & $10$ & 0.164 & 0.138 & 0.145 & 0.081 & 0.019 \\ 
 Bimodal & $20$ & 0.077 & 0.077 & 0.080 & 0.053 & 0.014 \\ 
 Bimodal & $100$ & 0.038 & 0.035 & 0.032 & 0.036 & 0.020 \\ 
Bimodal &  $1000$ & 0.019 & 0.021 & 0.019 & 0.020 & 0.017 \\ 
 \hline 
 Unimodal & $10$ & 0.444 & 0.451 & 0.425 & 0.410 & 0.432 \\ 
 Unimodal  & $20$ & 0.392 & 0.390 & 0.402 & 0.395 & 0.412 \\ 
Unimodal & $100$ & 0.280 & 0.266 & 0.272 & 0.280 & 0.297 \\ 
Unimodal  &  $1000$ & 0.098 & 0.102 & 0.091 & 0.104 & 0.113 \\ 
\hline
\end{tabular}
\end{table}

\section{Application to breast cancer genomics}
 \label{app}
 
  As a real data application, we consider classifying breast cancer (BRCA) tumor subtypes using data on microRNA (miRNA) abundance that are publicly available from the Cancer Genome Atlas (TCGA).  We use miRNA-Seq data for $n=1047$ BRCA tumor samples from different individuals, that was previously curated for a pan-cancer clustering application \citep{hoadley2018cell}.  The tumors are classified into one of $4$ canonical subtypes based on the PAM$50$ classifier for gene (mRNA) expression \citep{parker2009supervised}: Luminal A (LumA) ($n=572$), Luminal B (LumB) ($n=203$), Her2 $(n=83$), and Basal ($189$).  Previous exploratory analysis have shown that miRNA data shares some of the same BRCA subtype distinctions as gene expression \citep{park2019integrative,lock2013bayesian}.  Here, we examine the power of miRNAs to distinguish the subtypes using a supervised approach.   
  
  The miRNA expression read counts were log($1+x$)-transformed and centered to have mean $0$.  After transformation, we keep miRNAs with standard deviation greater than $0.5$, leaving $d=489$ miRNAs.  We consider separate classification tasks to discriminate each pair of subtypes, giving $6$ comparisons.  For each comparison, we apply (1) Bayesian DWD (BayesDWD) with uniform prior on $\lambda$, (2) DWD (i.e., the posterior mode) via the R package {\tt sdwd}, (3) SVM via the package {\tt e1071}, (4) random forest classification \citep{breiman2001random} via the package {\tt randomForest}, and (5) LDA via the package {\tt MASS}.  We assess each method by the average misclassification rate on the test sets under 10-fold cross-validation.  The results are shown in Table~\ref{testclass}.  Bayesian DWD, DWD, and SVM have similar classification performance across the different comparisons, while random forests and LDA perform less well by comparison.  For methods that give a class probability we consider the calibration of the estimated probabilities on the test sets vs.\ their empirical class proportions, as in Section~\ref{genmodsim}. Figure~\ref{fig:brca.calib} shows that the probabilities inferred for Bayesian DWD are overall well-calibrated.  The probabilities under the random forest  are also reasonable, but tend toward being over-conservative.  The probabilities for LDA are not well-calibrated. LDA suffers from over-fitting in situations with higher dimension and lower sample size; for this application, the vast majority (93\%) of predicted probabilities using LDA are below $0.02$ or above $0.98$, leaving limited data for the remaining bins.    
  
  \begin{table}
  \caption{ \label{testclass} Test classification rates for different subtype comparisons under cross validation using Bayesian DWD (BayesDWD), DWD, SVM, random forests (RF), and LDA. }
\centering
\begin{tabular}{c|ccccc}
  \hline
 Comparison & BayesDWD & DWD & SVM & RF & LDA \\ 
  \hline
LumA vs.\ LumB & 0.168 & 0.156 & 0.158 & 0.186 & 0.296 \\ 
  LumA vs.\ Her2 & 0.054 & 0.045 & 0.059 & 0.070 & 0.178 \\ 
  LumA vs.\ Basal & 0.026 & 0.028 & 0.025 & 0.030 & 0.084 \\ 
  LumB vs.\ Her2 & 0.129 & 0.124 & 0.155 & 0.169 & 0.188 \\ 
  LumB vs.\ Basal & 0.046 & 0.043 & 0.044 & 0.047 & 0.137 \\ 
  Her2 vs.\ Basal & 0.066 & 0.064 & 0.053 & 0.075 & 0.111 \\ 
   \hline
\end{tabular}
\end{table}

\begin{figure}[!ht]
\centering
\includegraphics[width=0.8\textwidth]{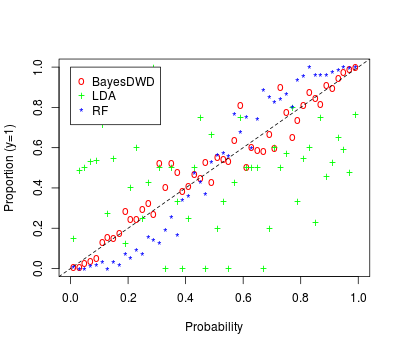}	
\caption{Estimated probabilities vs.\ observed proportions on test observations, aggregated across all comparisons, using Bayesian DWD, random forests (RF), or LDA.}
\label{fig:brca.calib}
\end{figure}

To illustrate the Bayesian DWD results, we focus on the comparison with the best classification accuracy from Table~\ref{testclass} (LumA  vs.\ Basal) and the comparison with the worst classification accuracy (LumA vs.\ LumB).  Figure~\ref{fig:cis_brca} plots the posterior mean of the DWD scores and their $95\%$ credible intervals for each comparison.  Note that the DWD scores are almost perfectly separated for LumA  vs.\ Basal, but not for LumA vs.\ LumB. The credible intervals for each plot help to understand the level of uncertainty in the discriminating projections.  

\begin{figure}[!ht]
\centering
\includegraphics[width=0.8\textwidth]{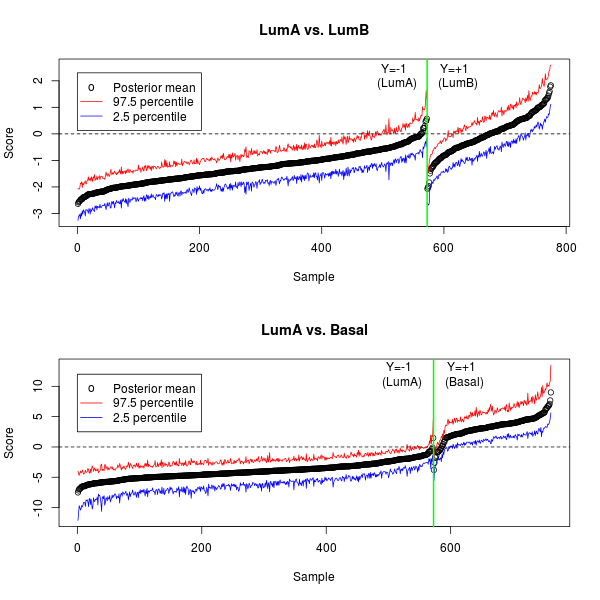}	
\caption{Mean DWD scores and associated 95\% credible intervals for LumA  vs.\ Basal and for LumA vs.\ LumB.}
\label{fig:cis_brca}
\end{figure}

\section{Discussion}
 \label{disc}
 
The Bayesian formulation of DWD is straightforward to modify or extend via  changes to the likelihood or prior.  As a future direction, we are interested in extending the model to accommodate multiple sources of data (e.g., multiple 'omics or clinical sources) or data with multiway structure \citep{lyu2016discriminating}.  In this article we have focused on the original and most widely used version of DWD, but the Bayesian framework may be extended for other versions including multi-class DWD \citep{huang2013multiclass}, weighted DWD \citep{qiao2010weighted}, sparse DWD \citep{wei2016direction}, and kernel DWD \citep{wang2018another,wang2019multicategory}.  Additionally, a generalized version of DWD was introduced in \citep{hall2005geometric}, wherein the distances $r_i$ in \eqref{dwd_obj} are raised to a power $q$.  The choice of $q$ modifies the loss function $V(\cdot)$ in \eqref{vu} \citep{wang2018another}, which in turn affects the probabilistic link \eqref{ydist}; the Bayesian formulation is potentially very useful in this context, because one can put a prior on $q$ and infer it within the model.          
 
Computing time and resources can limit a fully Bayesian treatment of DWD for some applications, as discussed in Section~\ref{computing}.  In such cases, the DWD solution can still be interpreted as a posterior mode to give point estimates for class probabilities, and the normal approximation in Section~\ref{asymp.norm} can be used to quickly assess uncertainty in the estimated coefficients.  Alternative computational approaches that facilitate a fully Bayesian treatment for massive datasets are worth pursuing. 

\acks{This work was supported by the National Institutes of Health (NIH) grant R01-GM130622.}

\appendix

\section{Proofs}
\setcounter{theorem}{0}
\begin{theorem}
If $\X \in \mathbb{R}^{d \times n}$, $\lambda \geq 0$, and $\y \in \{-1,1\}^n$ where $y_i=-1$ for some $i$ and $y_j=1$ for some $j$, then $p(\bbeta,\beta_0 \mid \X, \y, \lambda)$ in \eqref{bpost} gives a proper probability density function.   	
\end{theorem}

\begin{proof}
Without loss of generality assume $y_1=-1$ and $y_2=1$.  Define \[h(\bbeta,\beta_0 \mid \X, \y, \lambda) = \left[ \prod_{i=1}^n e^{-V\left(y_i(\beta_0 +\x_{i}^T \bbeta)\right)} \right]   e^{-(\lambda n/2) ||\bbeta||^2_2}.\]
By \eqref{bpost}, $p(\bbeta,\beta_0 \mid \X, \y, \lambda) \propto h(\bbeta,\beta_0 \mid \X, \y, \lambda)$. It suffices to show that (i) $h(\bbeta,\beta_0 \mid \X, \y, \lambda)$ is non-negative and (ii) $h(\bbeta,\beta_0 \mid \X, \y, \lambda)$ has finite integral over $\bbeta$ and $\beta_0$.  Condition (i) is trivially satisfied.  For condition (ii), note that $V(u)>0 \; \forall u\in\mathbb{R}$, where $V(u)$ is defined in \eqref{vu}.  Thus, 
\[e^{-V\left(y_i(\beta_0 +\x_{i}^T \bbeta)\right)}<1 \; \; \text{ for } i=1,\hdots,n.\]
Let $||\cdot||_1$ define the $l_1$-norm, and note that $||\X||_1||\bbeta||_1 > \mbox{max}\left(\{|\x_{i}^T \bbeta|: i=1,\hdots,n\}\right)$. For fixed $\bbeta \in \mathbb{R}^p$, 
\begin{align*}
  \int \prod_{i=1}^n e^{-V\left(y_i(\beta_0 +\x_{i}^T \bbeta)\right)} \, d\beta_0 	
 &= \int\displaylimits_{-||\X||_1||\bbeta||_1}^{||\X||_1||\bbeta||_1} \prod_{i=1}^n e^{-V\left(y_1(\beta_0 +\x_{1}^T \bbeta)\right)} \, d\beta_0 \\  & \; \; \; \; \; \; \; \; \; \; \; +\int\displaylimits_{-\infty}^{-||\X||_1||\bbeta||_1} \prod_{i=1}^n e^{-V\left(y_1(\beta_0 +\x_{1}^T \bbeta)\right)} \, d\beta_0 \\ & \; \; \; \; \; \; \; \; \; \; \; +\int\displaylimits_{||\X||_1||\bbeta||_1}^{\infty} \prod_{i=1}^n e^{-V\left(y_1(\beta_0 +\x_{1}^T \bbeta)\right)} \, d\beta_0 \\
 &< 2||\X||_1||\bbeta||_1+\int\displaylimits_{-\infty}^{-||\X||_1||\bbeta||_1}  e^{-V\left(y_2(\beta_0 +\x_{1}^T \bbeta)\right)} d\beta_0 \\  &\; \; \; \; \;  \;\;\;\;\;\;\;\;\;\;\;\;\;\;\;\;\;\;\;+\int\displaylimits_{||\X||_1||\bbeta||_1}^{\infty} e^{-V\left(y_1(\beta_0 +\x_{1}^T \bbeta)\right)} d\beta_0\\
 &< 2||\X||_1||\bbeta||_1+2 \int\displaylimits_{0}^{\infty} e^{-V(-u)}du \\
 &=2||\X||_1||\bbeta||_1+2 \int\displaylimits_{0}^{\infty} e^{1-u} du \\
 &=2||\X||_1||\bbeta||_1+2e.
\end{align*}
It follows that 
\begin{align*}
\int \int h(\bbeta,\beta_0 \mid \X, \y, \lambda) d\beta_0 d\bbeta &=  \int \left[ \int \prod_{i=1}^n e^{-V\left(y_i(\beta_0 +\x_{i}^T \bbeta)\right)} \, d\beta_0\right] e^{-(\lambda n/2) ||\bbeta||^2_2} d\bbeta \\
&< \int (2||\X||_1||\bbeta||_1+2e)e^{-(\lambda n/2) ||\bbeta||^2_2} d\bbeta \\
&=2||\X||_1 \int ||\bbeta||_1 e^{-(\lambda n/2) ||\bbeta||^2_2} d\bbeta +2e\int e^{-(\lambda n/2) ||\bbeta||^2_2} d\bbeta \\ 
&=2||\X||_1  \cdot p \left(\frac{\lambda n}{2}\right)\left(\frac{2 \pi}{\lambda n}\right)^{\frac{p-1}{2}}+2e\left(\frac{2 \pi}{\lambda n}\right)^{\frac{p}{2}}.
\end{align*} 
\end{proof}

\setcounter{theorem}{1}
\begin{theorem}
If $\X \in \mathbb{R}^{d \times n}$, $\lambda > 0$, and $\beta_0 \in  \mathbb{R}$,  then \[p(\bbeta \mid \X, \beta_0, \lambda) \propto A \cdot B\]
gives a proper probability density function, where $A$ and $B$ are defined in~\eqref{termA} and~\eqref{termB}.   	
\end{theorem}
\begin{proof}
Define $h(\bbeta,\beta_0, \X, \beta_0, \lambda) = A \cdot B$, where $A$ and $B$ are defined in~\eqref{termA} and~\eqref{termB}.  Note that $A$ and $B$ are always non-negative, and term $A$ is alway less than $1$.   Thus,
\begin{align*}
\int h(\bbeta,\beta_0, \X, \beta_0, \lambda) d \bbeta &< 	\int e^{-(\lambda n/2) ||\bbeta||^2_2} d\bbeta \\
&=\left(\frac{2 \pi}{\lambda n}\right)^{\frac{p}{2}}. 
\end{align*}

\end{proof}

\section{Algorithm}
\label{alg}

Here we describe in detail the algorithm to simulate draws $\{\bbeta^{(t)}, \beta_0^{(t)}, \lambda^{(t)}\}_{t=1}^T$ from their joint posterior distribution $p(\bbeta^{(t)}, \beta_0^{(t)}, \lambda^{(t)} \mid \X, \y)$. First, estimate the normalizing function $\hat{\phi}(\X,\lambda, \beta_0)$ as described in Section~\ref{lambda.inf}.  Then, initialize $\lambda^{(0)}$ to a positive value (e.g., $\lambda^{(0)}=1$) and initialize $\{\bbeta^{(0)}, \beta_0^{(0)}\}$ to the conditional posterior mode using the {\tt sdwd} package \citep{sparseDWD} or other software.  The algorithm is given below for iterations $t=1,\hdots,T$.  

\begin{enumerate}
\item For $i=1,\hdots,d$, update $\beta_i$ as follows:
\begin{enumerate}
\item Generate proposal $\beta_i^* \sim N(0,\sigma_*^2)$ 
\item Set $\tilde{\bbeta} = [\beta_1^{(t)} \dots \beta_{i-1}^{(t)} \; \beta_i^{(t-1)} \dots \beta_d^{(t-1)}]$ and $\bbeta^* = [\beta_1^{(t)} \dots \beta_{i-1}^{(t)} \; \beta_i^{*} \dots \beta_d^{(t-1)}]$
\item Compute $r=\frac{p(\bbeta^*,\beta_0^{(t-1)} \mid \X, \y, \lambda^{(t-1)})}{p(\tilde{\bbeta},\beta_0^{(t-1)} \mid \X, \y, \lambda^{(t-1)})}$, with the numerator and denominator given by \eqref{bpost}
\item Generate $u \sim \mbox{Uniform}(0,1)$
\item If $u<r$ set $\beta_i^{(t)}=\beta_i^*$, otherwise set $\beta_i^{(t)}=\beta_i^{(t-1)}$.
\end{enumerate}
\item Update $\beta_0$ as follows:
\begin{enumerate}
\item Generate proposal $\beta_0^* \sim N(0,0.25)$
\item Compute $r=\frac{p(\bbeta^{(t)},\beta_0^{*} \mid \X, \y, \lambda^{(t-1))}}{p(\bbeta^{(t)},\beta_0^{(t-1)} \mid \X, \y, \lambda^{(t-1)}}$, with the numerator and denominator given by \eqref{bpost}
\item Generate $u \sim \mbox{Uniform}(0,1)$
\item If $u<r$ set $\beta_0^{(t)}=\beta_0^*$, otherwise set $\beta_i^{(t)}=\beta_0^{(t-1)}$.
\end{enumerate}
\item Update $\lambda$ as follows:
\begin{enumerate}
\item Generate proposal $\lambda^{*}$ by $\lambda^*=e^{\mbox{log} \lambda^{(t-1)}+Z/4}$ where $Z \sim N(0,1)$  
\item  Compute $r=\frac{\lambda^* \hat{\phi}(\X,\lambda^*) e^{-0.5 n \lambda^* || \bbeta^{(t)} ||^2 }}{\lambda^{(t-1)} \hat{\phi}(\X,\lambda^{(t-1)}) e^{-0.5 n \lambda^{(t-1)} || \bbeta^{(t)} ||^2}}$
\item Generate $u \sim \mbox{Uniform} (0,1)$
\item If $u < r$ set $\lambda^{(t)}=\lambda^*$, otherwise set $\lambda^{(t)}= \lambda^{(t-1)}$.
\end{enumerate} 
\item Set $\sigma^2_*$ to half of the sample variance of $\{\beta_i\}_{i=1}^d$.
\end{enumerate}

Note that if $\lambda$ is fixed, step (c) above can be skipped. The adaptive updating of the proposal variance in step (d) is not necessary for the validity of the algorithm, but can improve mixing considerably when the approximate scale of the coefficients $\bbeta$ is uncertain.

\bibliography{bibliography.bib}

\end{document}